\documentclass[aps,pra,nobibnotes]{revtex4}

\usepackage{fancyhdr}

\fancypagestyle{pg}
{%
\lhead{}
\rhead{}
\cfoot{--\ \thepage\ --}

}

\usepackage{amsmath}
\usepackage{amsthm} 
\usepackage{latexsym}
\usepackage{amsfonts}
\usepackage{amssymb}
\usepackage{color}
\usepackage{bbm,dsfont}
\usepackage{graphicx}
\usepackage{hyperref}
\newtheorem{proposition?}{Proposition?}
\newtheorem{theorem}{Theorem}
\newtheorem{lemma}{Lemma}

\newtheorem{definition}{Definition}

\newcommand{\hi}{\mathcal{H}} %Hilbert space H
\newcommand{\hik}{\mathcal{K}} %Hilbert space K
\newcommand{\lh}{\mathcal{L(H)}} %bounded linear operators
\newcommand{\lk}{\mathcal{L(K)}} %bounded linear operators on K
\newcommand{\id}{\mathbf{1}} %identity operator 2
  \makeatletter
  \def\mathcomposite{%
     \@ifstar
        {\def\@mathcomposite@option{%
            \baselineskip\z@skip\lineskiplimit-\maxdimen}%
         \@mathcomposite}%
        {\let\@mathcomposite@option\offinterlineskip
         \@mathcomposite}}
  \def\@mathcomposite{%
     \@ifnextchar[\@@mathcomposite{\@@mathcomposite[0]}}
  \def\@@mathcomposite[#1]#2#3#4{%
     #2{\mathchoice
        {\@mathcomposite@{#1}{#3}{#4}\displaystyle{1}}%
        {\@mathcomposite@{#1}{#3}{#4}\textstyle{1}}%
        {\@mathcomposite@{#1}{#3}{#4}%
         \scriptstyle\defaultscriptratio}%
        {\@mathcomposite@{#1}{#3}{#4}%
         \scriptscriptstyle\defaultscriptscriptratio}}}
  \def\@mathcomposite@#1#2#3#4#5{%
     \vcenter{\m@th\@mathcomposite@option
        \dimen@\f@size\p@\dimen@#1\dimen@\dimen@#5\dimen@
        \divide\dimen@ 18
        \edef\@mathcomposite@skipamount{\the\dimen@}%
        \ialign{\hfil$#4##$\hfil\cr
           #2\crcr
           \noalign{\vskip\@mathcomposite@skipamount}%
           #3\crcr}}}
  \makeatother

\usepackage{framed}

\usepackage{comment}
\usepackage{xcolor}

\begin{document}

\title[]{
%The compatibility of quantum channels depends on the construction of composite systems
Incompatibility of quantum channels in general
probabilistic theories
}
\thanks{MY and TM contributed equally to this work.}

%\author{} 
\author{Masataka Yamada}
\email{yamada.masataka.42z@st.kyoto-u.ac.jp}
\affiliation{
Department of Nuclear Engineering, Kyoto University, 6158540 Kyoto, Japan}
\author{Takayuki Miyadera}
\email{miyadera@mi.meijigakuin.ac.jp}
%\thanks{
%From April 2024: 
\affiliation{
Institute for Mathematical Informatics, Meiji Gakuin University, 2448539 Yokohama, Japan}
% \\
%(miyadera@mi.meijigakuin-u.ac.jp)
%}
%\email{
%\affiliation{
%Department of Nuclear Engineering, Kyoto University, 6158540 %Kyoto, Japan}
%\altaffiliation{
%\affiliation{

\begin{abstract}

In quantum theory, there exist sets of operations that cannot be performed simultaneously. 
These sets of operations are referred to as incompatible. While this definition of 
incompatibility extends to general probabilistic theories (GPTs), 
the dependency of the set of compatible sets on 
the definition of composite systems has not been 
thoroughly investigated. 
For quantum channels, compatibility is defined using the tensor product of Hilbert spaces, 
based on the conventional composite system. However, in GPTs, composite systems are not
 uniquely defined, and the set of states can vary from the minimal tensor to the maximal
 tensor.
In this paper, 
in addition to 
the usual 
quantum compatibility, 
we introduce min-tensor-compatibility using the minimal 
tensor on the composite 
system of effect spaces 
and investigate their relationship employing noisy 
identity channels on qubits. As a result, we found that
the set of min-tensor-compatible channel pairs is strictly broader 
than the set of quantum-compatible channel pairs. 
Furthermore, we introduce the concept of almost quantum compatible pairs of channels 
from an operational perspective. This concept corresponds 
to cases where the correlation functions in the verification of 
compatibility can be realized through a channel and 
local reinterpretation of effects. We demonstrate that 
the set of all almost quantum compatible channel pairs is strictly narrower 
than the set of all min-tensor-compatible channel pairs.
\end{abstract}

%\pacs{03.65.Ta}

\maketitle
%%%%%%%%%

\section{Introdocution}

In the realm of quantum theory, there exist pairs of operations that cannot be carried out 
simultaneously, a phenomenon termed incompatibility\cite{invi}. 
One of the most prominent examples is the uncertainty principle\cite{Heisenberg, Werner, POVMuncertain, book}, 
which asserts that certain pairs of observables cannot be measured concurrently. 
This concept extends beyond observables 
to encompass quantum channels as well\cite{incompchannel,
maxcomp}; 
certain pairs of channels cannot be realized simultaneously. 
The principle is exemplified 
by the no-cloning theorem\cite{cloning1, cloning2}, which highlights the incompatibility of two identity channels.

While incompatibility serves to underscore the disparities between quantum and 
classical theories, its relevance extends to broader physical frameworks. 
Indeed, the no-cloning theorem holds in general probabilistic theories (GPTs) 
beyond classical theory\cite{cloningGPT}, a fact elucidated within the context of incompatibility. 
Investigating these generalized frameworks holds promise for uncovering 
key characteristics of quantum theory.

This paper delves into the study of incompatibility of 
channels, focusing particularly on 
the principles governing the construction of composite systems. 
%%%%%%%%%
%%%%%%%%

In GPT, although there has been research characterizing single quantum systems 
based on physical or information-theoretical criteria\cite{
MasanesMuller,DAriano,Thirdorder}, there is not a unique solution 
for determining the composite system for a given pair of systems. 
This leads to a variety of potential composite systems beyond the conventional 
quantum framework\cite{GPTreview, 
takakura, Plavala}. In fact, natural requirements for composite systems only define 
the maximum and minimum state spaces, neither of which aligns with the true quantum 
composite system. Additionally, the maximum value related to the CHSH inequality, 
a key measure characterizing composite systems, does not differentiate between 
the maximum and quantum state spaces. Therefore, it is intriguing to identify 
any phenomena that vary among the different composite systems. 
The definition of generalized incompatibility inherently relies on how 
the composite system is constructed. This raises the question of whether 
a pair of channels deemed incompatible within the standard quantum composite system
 might show compatibility within alternative composite systems.

More specifically, we demonstrate that pairs of noisy identity channels 
can be compatible within a min-tensored composite system despite being 
incompatible within the normal quantum composite system. 
This disparity prompts the introduction of new categories of compatible channels. 
Of particular interest is the exploration of whether this gap can be bridged by 
considering compatibility in a device-independent manner.
\par
The paper is organized as follows: 
Section \ref{sec:prelim} begins with an exploration of preliminary concepts surrounding quantum incompatibility and properties of composite systems within general probabilistic theories. Within this context, we introduce the notion of min-tensor-compatibility, which is grounded in identifying composite systems that yield the smallest effect spaces.
Section \ref{sec:main} presents the primary findings. Initially, we contrast the normal quantum compatibility class with the min-tensor-compatibility class by analyzing a pair of noisy identity channels. Additionally, we introduce a distinct class of compatibility called almost quantum compatibility, motivated by operational considerations. This class is then compared with min-tensor-compatibility.
%%%%%%%%%%%%%%%%%%%%%%%
\section{Preliminaries}\label{sec:prelim}

\subsection{quantum incompatibility of channels}

A quantum system is described by a Hilbert space. 
In the following, Hilbert spaces are assumed finite-dimensional. 
The set of all linear (bounded) operators acting in a Hilbert space $\mathcal{H}$ 
 is denotedy by  $\mathcal{L} (\mathcal{H}) $. 
 Its subset 
 $\mathcal{L}_+ (\mathcal{H}) $ is the set of all positive operators in $\mathcal{L} (\mathcal{H}) $, which forms a positive cone 
 in $\mathcal{L}_s(\hi):=\{A|\ A\in \lh, A=A^*\}$, 
 the set of all self-adjoint operators. 
 The positive cone completely specifies the theory (system). 
The positive cone is used to introduce 
%We further introduce 
$\mathcal{E}(\hi):=
\mathcal{L}_+(\hi) \cap (\id - \mathcal{L}_+(\hi))=
\{E|\ E\in \mathcal{L}_+(\hi), E\leq \id\}$, the set of all 
effects. The state space $\mathcal{S}(\hi)$ is 
a subset of its dual 
as $\mathcal{S}(\hi):=
\{\rho|\ 0\leq \langle \rho, E\rangle \leq 1\mbox{ for all }
E\in \mathcal{E}(\hi), \langle \rho, \id\rangle =1\}$, 
which can be identified with the set of all 
density operators through Hilbert-Schmidt inner product
$\langle \rho, E\rangle = \mbox{tr}[\rho E]$. 
A physical state change is described by a map 
called a channel. 
A linear map $\Lambda: \lk \to \lh$ is called a channel 
if it is unit-preserving and completely positive. 
(In the following, we will mainly use the Heisenberg picture. A superscript $\ ^*$ refers to the Schrödinger picture.) 
The complete positivity means that a map  
$\Lambda \otimes id: \lk \otimes \mathcal{L}(
\mathbb{C}^d) \to \lh \otimes \mathcal{L}(\mathbb{C}^d)$ 
is positive for all $d<\infty$. 
%%%%%%%%%%%%
%%%%%%%%%%%%%%%%
 \par
 In quantum theory, a composite system of 
 two systems described by Hilbert spaces $\hi_1$ and $\hi_2$ 
 is described by a tensored Hilbert space $\hi_1 \otimes \hi_2$. 
%Thus we arrive at the 
The following definitions of 
compatibility and incompatibility\cite{invi} are standard. 
%%%%%%%%%%%%%%%%%%
\begin{definition}
Suppose that we have three quantum systems 
described by Hilbert spaces $\hi_0$, $\hi_1$ and $\hi_2$. 
Channels $\Lambda_1: \mathcal{L}(\hi_1) \to \mathcal{L}(\hi_0)$ 
and $\Lambda_2: \mathcal{L}(\hi_2) \to \mathcal{L}(\hi_0)$
are {\bf quantum compatible} if 
there exists a channel $\Lambda: \mathcal{L}
(\hi_1 \otimes \hi_2) \to \mathcal{L}(\hi_0)$ 
satisfying 
$\Lambda(X\otimes \id) = \Lambda_1(X)$ and 
$\Lambda(\id \otimes Y) = \Lambda_2(Y)$ 
for all $X\in \mathcal{L}(\hi_1)$ and $Y \in \mathcal{L}(\hi_2)$. 
The channel $\Lambda$ is called a joint channel.  
We call $\Lambda_1$ and $\Lambda_2$ 
{\bf quantum incompatible}
if they are not quantum compatible.  
\end{definition}  
Note that for compatible $\Lambda_1$ and $\Lambda_2$ their joint channel 
is not unique in general. 
\subsection{min-tensor-incompatibility of channels}
%%%%%%%%%
In GPT, a composite system is not uniquely determined 
for a given pair of systems\cite{GPTreview, 
takakura, Plavala}. For a given pair of quantum systems 
described by $\hi_1$ and $\hi_2$, 
the state space of a composite system 
is a convex subset $\mathcal{S}$ of a set 
$\mathcal{S}(\hi_1)\otimes_{max}
\mathcal{S}(\hi_2)$ defined by  
\begin{eqnarray*}
&&\mathcal{S}(\hi_1)\otimes_{max}
\mathcal{S}(\hi_2)
\\
&&
:=\{\omega|\ 
\omega: \mathcal{L}(\hi_1) \times 
\mathcal{L}(\hi_2) \to \mathbb{C}  
\mbox{ bilinear map satisfying } 0 \leq \omega(E,F) \leq 1 \mbox{ for all } 
E\in \mathcal{E}(\hi_1), F\in \mathcal{E}(\hi_2),   
\mbox{ and }\omega(\id, \id) =1 \}.
\end{eqnarray*}
$\mathcal{S}$ satisfies 
$\mathcal{S}(\hi_1)\otimes_{min}\mathcal{S}(\hi_2)
\subset \mathcal{S}\subset \mathcal{S}(\hi_1)
\otimes_{max}\mathcal{S}(\hi_2)$, where 
\begin{eqnarray*}
&&\mathcal{S}(\hi_1)\otimes_{min}\mathcal{S}(\hi_2)
\\
&&
:=\{\omega |\ \omega =\sum_n \lambda_n 
\sigma_n \otimes \eta_n, 0 \leq \lambda_n \leq 1, 
\sum_n \lambda_n =1, \sigma_n \in \mathcal{S}(\hi_1), 
\eta_n \in \mathcal{S}(\hi_2)\}. 
\end{eqnarray*}
Note that $\mathcal{S}(\hi_1\otimes \hi_2)$ 
coincides with neither of $\mathcal{S}(\hi_1) 
\otimes_{min}\mathcal{S}(\hi_2)$
and $\mathcal{S}(\hi_1) \otimes_{max}
\mathcal{S}(\hi_2)$. Thus we need another 
condition to specify the quantum composite system. 
Positive cones which generate effect spaces 
and thus their corresponding state spaces 
are denoted by 
$\mathcal{L}_+(\hi_1)\otimes_{max}
\mathcal{L}_+(\hi_2)$ for the state space $\mathcal{S}(\hi_1)
\otimes_{min}\mathcal{S}(\hi_2)$ 
and $\mathcal{L}_+(\hi_1) \otimes_{min}
\mathcal{L}_+(\hi_2)$ for the state space $
\mathcal{S}(\hi_1) \otimes_{max}\mathcal{S}(\hi_2)$. 
For state spaces $\mathcal{S}_1\subset \mathcal{S}_2$, 
their corresponding positive cones $\mathcal{P}_1$ 
and $\mathcal{P}_2$  
satisfy $\mathcal{P}_1 \supset \mathcal{P}_2$ in general. 
In particular, the concrete form of $
\mathcal{L}_+(\hi_1)\otimes_{min} \mathcal{L}_+(\hi_2)$ 
is  
\begin{eqnarray*}
\mathcal{L}_+(\hi_1)\otimes_{min} \mathcal{L}_+(\hi_2)
=\{E\in \mathcal{L}_+(\hi_1\otimes \hi_2)|\ 
E=\sum_n \lambda_n F_n\otimes G_n, 
\lambda_n \geq 0, F_n \in \mathcal{L}_+(\hi_1), 
G_n \in \mathcal{L}_+(\hi_2)\}. 
\end{eqnarray*}
For each composite system, it is natural to define its corresponding 
compatibility. 
%%%%%%%%%%
\begin{definition}
Suppose that we have three quantum systems 
described by Hilbert spaces $\hi_0$, $\hi_1$ and $\hi_2$. 
We consider a composite system of 
$\mathcal{L}_+(\hi_1)$ and $\mathcal{L}_+(\hi_2)$ 
defined by its positive cone $\mathcal{P}$. 
Channels $\Lambda_1: \mathcal{L}(\hi_1) \to \mathcal{L}(\hi_0)$ 
and $\Lambda_2: \mathcal{L}(\hi_2) \to \mathcal{L}(\hi_0)$
are {\bf $\mathcal{P}$-compatible} if 
there exists a linear map $\Lambda: \mathcal{L}
(\hi_1 \otimes \hi_2) \to \mathcal{L}(\hi_0)$ 
satisfying 
$\Lambda(\mathcal{P}) \subset \mathcal{L}_+(\hi_0)$ and 
$\Lambda(X\otimes \id) = \Lambda_1(X)$ and 
$\Lambda(\id \otimes Y) = \Lambda_2(Y)$ 
for all $X\in \mathcal{L}(\hi_1)$ and $Y \in \mathcal{L}(\hi_2)$.
The map $\Lambda$ is called a joint map.  
We call $\Lambda_1$ and $\Lambda_2$ 
{\bf $\mathcal{P}$-incompatible}
if they are not $\mathcal{P}$-compatible.  
In particular, we call $
\mathcal{L}_+(\hi_1)\otimes_{min}
\mathcal{L}_+(\hi_2)$-(in)compatible 
as {\bf min-tensor-(in)compatible}. 
%%%%%%%%%%%
%%%%%
\end{definition}
From the definition, it is trivial to see that 
for two composite systems specified by 
positive cones $\mathcal{P}_1 \subset 
\mathcal{P}_2$ a pair of $\mathcal{P}_2$-compatible channels
$(\Lambda_1, \Lambda_2)$ 
is $\mathcal{P}_1$-compatible. 
In particular, quantum compatible $\Lambda_1$ and 
$\Lambda_2$ are min-tensor-compatible. 
%%%%%%%%%%%%%%%
\section{Main results}\label{sec:main}
\subsection{quantum compatibility vs min-tensor-compatibility}
In this section, we study the min-tensor-compatibility of 
two noisy identity channels to compare the 
min-tensor-compatibility and the quantum compatibility. 
We treat a family of noisy identity channels
 $\Lambda_{\eta}: \mathcal{L}(\mathbb{C}^2) \to \mathcal{L}(\mathbb{C}^2)$ 
$(0\leq \eta \leq 1)$  
on a qubit defined by, for all $A\in 
\mathcal{L}(\mathbb{C}^2)$, 
\begin{eqnarray*}
\Lambda_{\eta}(A)= \eta A + \frac{1-\eta}{2}\mbox{tr}[A]
\id. 
\end{eqnarray*}
The channel is a probabilistic mixture of the identity channel and 
the completely depolarized channel.
The (quantum) no-cloning theorem represents 
the quantum incompatibility of $\Lambda_1$ (i.e., $\eta=1$)
 and $\Lambda_1$ itself.  
It is known that $\Lambda_{\eta_1}$ and $\Lambda_{\eta_2}$ are quantum compatible if and only if 
the parameters $\eta_1$ and $\eta_2$ satisfy\cite{optimalclone}
\begin{eqnarray}
\eta_1^2 + \eta_2^2 + (1-\eta_1 - \eta_2)^2 \leq 1. 
\label{eq:qoptimal}
\end{eqnarray}
On the other hand, the generalized no-cloning theorem\cite{cloningGPT}
is applied to show $\Lambda_1$ and $\Lambda_1$ 
(i.e., $\eta_1=\eta_2=1$) 
are min-tensor-incompatible. 
We find the following theorem for the min-tensor-compatibility. 
%%%%%%%%%%%%%
%%%%%%%%%%%%%%%%%%%%%%
\begin{theorem}\label{theorem1}
 A pair of noisy identity channels $\Lambda_{\eta_1}$ and $\Lambda_{\eta_2}$ is min-tensor-compatible 
 if and only if 
 \begin{eqnarray}
 \eta_1^2 + \eta_2^2 \leq 1.
 \end{eqnarray}
 In addition for $\eta_1, \eta_2$ satisfying $\eta_1^2 + \eta_2^2 =1$, the joint map $\Lambda$ is unique.
\end{theorem}
%%%%%%%%%%%%%%%
\begin{proof}
 Let us assume that $\Lambda_{\eta_1} $ 
 and $\Lambda_{\eta_2}$ are 
 min-tensor-compatible and $\Lambda$ is a joint map. 
 For arbitrary normal vectors $\mathbf{n}$ and $\mathbf{m}$ in $\mathbb{R}^3$, 
 the joint map defines a POVM 
 $\{\Lambda(\frac{1}{2}(\id \pm \mathbf{n} \cdot \boldsymbol\sigma )
 \otimes \frac{1}{2}(\id \pm \mathbf{m} \cdot \boldsymbol\sigma))\}$. 
 The POVM is a joint POVM of 
 %%%   
 \begin{eqnarray*}
 \{ \Lambda((\id \pm \mathbf{n}\cdot \boldsymbol\sigma)/2 \otimes \id)\}=\{\Lambda_{\eta_1}
 ((\id \pm \mathbf{n}\cdot \boldsymbol\sigma)/2) 
 \}
 =\{ (\id \pm \eta_1 \mathbf{n}\cdot \boldsymbol\sigma)/2\}
 \end{eqnarray*}
 and
 \begin{eqnarray*}
 \{\Lambda( \id \otimes (\id \pm \mathbf{m}\cdot \boldsymbol\sigma)/2)\}
 =\{\Lambda_{\eta_2}( (\id \pm \mathbf{m}\cdot \boldsymbol\sigma)/2)\}
 =
 \{(\id \pm \eta_2 \mathbf{m}\cdot \boldsymbol\sigma)/2\}
 \end{eqnarray*}
 These unbiased qubit observables are jointly measurable if and only if 
 $|\eta_1 \mathbf{n}|^2 + |\eta_2 \mathbf{m}|^2 
 \leq 1 + (\eta_1 \eta_2 (\mathbf{n}\cdot \mathbf{m}))^2$ holds\cite{spinuncertain}.
 Taking orthogonal $\mathbf{n}$ and $\mathbf{m}$, we conclude 
 \begin{eqnarray*}
 \eta_1^2 + \eta_2^2 \leq 1. 
 \end{eqnarray*} 
Now we show that $\Lambda_{\eta_1}$ and $\Lambda_{\eta_2}$ 
are min-tensor-compatible if $\eta_1$ and $\eta_2$ satisfy the above inequality. 
Let us first consider the case $\eta_1^2 + \eta_2^2 =1$. 
We construct a joint map of $\Lambda_{\eta_1}$ and $\Lambda_{\eta_2}$, 
which is also shown to be the unique joint map. 
We set an orthogonal basis $\{\mathbf{e}_1, \mathbf{e}_2, \mathbf{e}_3\}$ of $\mathbb{R}^3$. 
A set $\{\xi_0, \xi_1, \xi_2, \xi_3, \xi_4\}:=
\{\id, \mathbf{e}_1\cdot \boldsymbol\sigma, \mathbf{e}_2 \cdot \boldsymbol\sigma, \mathbf{e}_3 \cdot \boldsymbol\sigma\}$ 
forms a basis of $\mathcal{L}(\mathbb{C}^2)$. 
%%%%%%%%
%%%%%%%
A linear map 
$\Lambda: \mathcal{L}(\mathbb{C}^2) \otimes\mathcal{L}(\mathbb{C}^2) \to \mathcal{L}(\mathbb{C}^2)$ is completely specified by coefficients $\{\Lambda^{\xi}_{\mu \nu}\}$ defined by 
\begin{eqnarray*}
\Lambda( \xi_{\mu} \otimes \xi_{\nu}) = \sum_{\alpha=0}^3 
\Lambda^{\alpha}_{\mu \nu}\xi_{\alpha}. 
\end{eqnarray*}
 For the map $\Lambda$ to be a joint map of $\Lambda_{\eta_1}$ and 
 $\Lambda_{\eta_2}$, the coefficents must satisfy for $i, j =1,2,3$, 
 \begin{eqnarray*}
 \Lambda_{i 0}^{\alpha} = \delta_{i \alpha} \eta_1
 \\
 \Lambda_{0 j}^{\alpha} = \delta_{j \alpha} \eta_2
 \end{eqnarray*}
 and $\Lambda_{00}^{\alpha}= \delta_{0 \alpha}$. 
 For $\Lambda$ to be a well-defined map from $\mathcal{L}_+(\mathbb{C}^2)
 \otimes_{min}\mathcal{L}_+(\mathbb{C}^2)$
 to $\mathcal{L}_+(\mathbb{C}^2)$, for all $i,j=1,2,3$ the 
 following inequality must hold:  
 \begin{eqnarray*}
 \Lambda( (\xi_0 \pm \xi_i)\otimes (\xi_0 \pm \xi_j))\geq 0. 
 \end{eqnarray*}
 That is, 
 \begin{eqnarray*}
 \id + \eta_1 \xi_i + \eta_2 \xi_j + \sum_{\alpha} \Lambda^{\alpha}_{ij} \xi_\alpha \geq 0
 \\
  \id - \eta_1 \xi_i + \eta_2 \xi_j - \sum_{\alpha} \Lambda^{\alpha}_{ij} \xi_\alpha \geq 0
\\
 \id + \eta_1 \xi_i - \eta_2 \xi_j - \sum_{\alpha} \Lambda^{\alpha}_{ij} \xi_\alpha \geq 0
 \\
  \id - \eta_1 \xi_i - \eta_2 \xi_j + \sum_{\alpha} \Lambda^{\alpha}_{ij} \xi_\alpha \geq 0. 
 \end{eqnarray*}
 For $i=1$ and $j=2$, they give 
 \begin{eqnarray*}
 1+\Lambda^0_{12} \geq \sqrt{ (\eta_1 + \Lambda^1_{12})^2 + 
 (\eta_2 + \Lambda^2_{12})^2 +(\Lambda^3_{12})^2}=\sqrt{2(\Lambda^1_{12} \eta_1 + \Lambda^2_{12}\eta_2) 
 +1 + (\Lambda^1_{12})^2 + (\Lambda^2_{12})^2 + (\Lambda^3_{12})^2}\\
   1- \Lambda^0_{12} \geq \sqrt{ (\eta_1 + \Lambda^1_{12})^2 + 
 (\eta_2 - \Lambda^2_{12})^2 +(\Lambda^3_{12})^2}=\sqrt{2(\Lambda^1_{12} \eta_1 - \Lambda^2_{12}\eta_2) 
 +1 + (\Lambda^1_{12})^2 + (\Lambda^2_{12})^2 + (\Lambda^3_{12})^2}\\
    1- \Lambda^0_{12} \geq \sqrt{ (\eta_1 - \Lambda^1_{12})^2 + 
 (\eta_2 + \Lambda^2_{12})^2 +(\Lambda^3_{12})^2}=\sqrt{-2(\Lambda^1_{12} \eta_1 - \Lambda^2_{12}\eta_2) 
 +1 + (\Lambda^1_{12})^2 + (\Lambda^2_{12})^2 + (\Lambda^3_{12})^2}\\
 1+\Lambda^0_{12} \geq \sqrt{ (\eta_1 -\Lambda^1_{12})^2 + 
 (\eta_2 - \Lambda^2_{12})^2 +(\Lambda^3_{12})^2}=\sqrt{-2(\Lambda^1_{12} \eta_1 + \Lambda^2_{12}\eta_2) 
 +1 + (\Lambda^1_{12})^2 + (\Lambda^2_{12})^2 + (\Lambda^3_{12})^2}. 
 \end{eqnarray*}
 Let us first assume $\Lambda^0_{12}>0$. This contradicts with the second inequality in the case of 
 $\Lambda^1_{12}\eta_1 - \Lambda^2_{12}\eta_2 \geq 0$ and 
 with the third one in the case of $\Lambda^1_{12}\eta_1 -\Lambda^2_{12}\eta_2 <0$. 
 Therefore $\Lambda^0_{12}>0$ cannot be true. On the other hand, $\Lambda^0_{12}<0$ 
 contradicts the first and the fourth inequalities in a similar way. Thus we conclude 
 $\Lambda^0_{12}=0$ and $\Lambda^1_{12}=\Lambda^2_{12}=\Lambda^3_{12}=0$ for the above 
 inequalities to hold. 
 %The same argument holds for other pairs satisfying $i \neq j$. 
 For any normalized
$\mathbf{e}$ and $\mathbf{d}$
satisfying $\mathbf{e}\cdot \mathbf{d}=0$, 
we can set an orthonormalized basis 
$\{\mathbf{e}_1, \mathbf{e}_2, \mathbf{e}_3\}$ 
so that  
$\mathbf{e}_1=\mathbf{e}$ and $\mathbf{e}_2 = \mathbf{d}$. 
 Thus we conclude 
for any $\mathbf{e}\cdot \mathbf{d}=0$, 
\begin{eqnarray*}
\Lambda(\mathbf{e}\cdot \boldsymbol\sigma \otimes \mathbf{d}\cdot \boldsymbol\sigma)=0. 
\end{eqnarray*}
%%%%%%%%%%
Let us examine $\Lambda(\mathbf{n}\cdot \boldsymbol\sigma\otimes \mathbf{m}\cdot \boldsymbol\sigma)$ 
for arbitrary normalized vectors $\mathbf{n}$ and $\mathbf{m}$. 
%%%%%%%%%
%%%%%%%%%%%%%
%%%%%%%%%%%%
As $\mathbf{m}$ is written as $\mathbf{m} = (\mathbf{n}\cdot \mathbf{m}) \mathbf{n} + 
(\mathbf{m} - (\mathbf{n}\cdot \mathbf{m}) \mathbf{n})$ with $(\mathbf{m} - (\mathbf{n}\cdot 
\mathbf{m})\mathbf{n}) \cdot \mathbf{n}=0$,  we obtain  
\begin{eqnarray*}
\Lambda(\mathbf{n} \cdot \boldsymbol\sigma \otimes \mathbf{m} \cdot \boldsymbol\sigma)
= (\mathbf{n}\cdot \mathbf{m}) \Lambda(\mathbf{n}\cdot \boldsymbol\sigma 
\otimes \mathbf{n} \cdot \boldsymbol\sigma). 
\end{eqnarray*}
Similarly it holds that $\Lambda( \mathbf{n} \cdot \boldsymbol\sigma 
\otimes \mathbf{m} \cdot \boldsymbol\sigma)= 
(\mathbf{n} \cdot \mathbf{m}) \Lambda(\mathbf{m} \cdot 
\sigma \otimes \mathbf{m} \cdot \boldsymbol\sigma)$. 
Therefore we find for any $\mathbf{n} \cdot \mathbf{m} \neq 0$, 
\begin{eqnarray*}
\Lambda(\mathbf{n} \cdot \boldsymbol\sigma \otimes \mathbf{n} \cdot 
\boldsymbol\sigma) = \Lambda (\mathbf{m} \cdot \boldsymbol\sigma 
\otimes \mathbf{m} \cdot \boldsymbol\sigma). 
\end{eqnarray*}
Thus $N:= \Lambda(\mathbf{n} \cdot \boldsymbol\sigma \otimes 
\mathbf{n} \cdot \boldsymbol\sigma)$ is independent of $\mathbf{n}$ and is written as 
\begin{eqnarray*}
N=a_0 \id + \mathbf{a} \cdot \boldsymbol\sigma 
\end{eqnarray*}
with $(a_0, \mathbf{a})\in \mathbb{R}^4$. Now for $\Lambda$ to be well-defined, 
it must hold that for arbitrary $\mathbf{n}, \mathbf{m}$, 
\begin{eqnarray*}
\Lambda\left( \frac{1}{2}(\id + \mathbf{n}\cdot \boldsymbol\sigma)
\otimes \frac{1}{2}(\id + \mathbf{m}\cdot \boldsymbol\sigma)\right)
= \frac{1}{4}
\left(
\id + (\eta_1 \mathbf{n} + \eta_2 \mathbf{m})\cdot \boldsymbol\sigma 
+ (\mathbf{n}\cdot \mathbf{m}) 
(a_0 \id + \mathbf{a}\cdot \boldsymbol\sigma)
\right) \geq 0.
\end{eqnarray*}
%%%%%%%%%%%
Let us assume $\mathbf{n} \cdot \mathbf{m} \neq \pm 1$. 
We introduce a normalized vector 
$\mathbf{r}:= \frac{\eta_1 \mathbf{n} + \eta_2 \mathbf{m}}
{\sqrt{1 + 2 (\mathbf{n} \cdot \mathbf{m}) \eta_1 \eta_2}}$ to rewrite the above quantity as 
\begin{eqnarray}
&&\frac{1}{4}
\left(
\id + (\eta_1 \mathbf{n} + \eta_2 \mathbf{m})\cdot \boldsymbol\sigma 
+ (\mathbf{n}\cdot \mathbf{m}) 
(a_0 \id + \mathbf{a}\cdot \boldsymbol\sigma)
\right)
\nonumber \\
&=& 
\frac{1}{4}
\left(
\id + \sqrt{
1+ 2\eta_1 \eta_2 (\mathbf{n}\cdot \mathbf{m})}
\mathbf{r}\cdot \boldsymbol\sigma 
+ (\mathbf{n}\cdot \mathbf{m}) 
(a_0 \id + \mathbf{a}\cdot \boldsymbol\sigma)
\right)
\nonumber 
\\
&=&
\frac{1}{4}
\left(
(1+ \mathbf{n}\cdot \mathbf{m} a_0)\id 
+
\left(
\sqrt{1+ 
2 (\mathbf{n}\cdot \mathbf{m}) \eta_1 \eta_2 } 
\mathbf{r}
+ (\mathbf{n}\cdot \mathbf{m})\mathbf{a}
\right) \cdot \boldsymbol\sigma
\right)\geq 0. 
\label{IneqFund}
\end{eqnarray}
Let us set $\mathbf{n}$ and $\mathbf{m}$ so that 
$\mathbf{n}\cdot \mathbf{a}=\mathbf{m}\cdot \mathbf{a}
=0$. 
The condition gives 
$\mathbf{r} \cdot \mathbf{a}=0$. 
We obtain 
\begin{eqnarray*} 
1+ \mathbf{n} \cdot \mathbf{m} a_0 
\geq \sqrt{ 1+ 2(\mathbf{n}\cdot \mathbf{m})\eta_1 \eta_2 
+ (\mathbf{n}\cdot \mathbf{m})^2 |\mathbf{a}|^2}. 
\end{eqnarray*}
Therefore it must hold that
\begin{eqnarray}
1+ 2 \mathbf{n} \cdot \mathbf{m} a_0 
+(\mathbf{n}\cdot \mathbf{m})^2 a_0^2 
 \geq  1+ 2 \mathbf{n} \cdot \mathbf{m} \eta_1 \eta_2
 +(\mathbf{n}\cdot \mathbf{m})^2 |\mathbf{a}|^2
 \geq 
 1+ 2 \mathbf{n} \cdot \mathbf{m} \eta_1 \eta_2. 
\label{ineqa0}
\end{eqnarray}
As the vectors $\mathbf{n}$ and $\mathbf{m}$ 
can be chosen arbitrarily as far as $\mathbf{n}\cdot 
\mathbf{a}
= \mathbf{m}\cdot \mathbf{a}=0$, 
the value $\mathbf{n}\cdot \mathbf{m}$ 
can take any value in $(-1, 1)$.
%%%%%%%
%%%%%%%%%%%%%%%%%%%%
\par
%%%%%%%%
For $\mathbf{n}\cdot \mathbf{m}>0$, 
the above inequality (\ref{ineqa0}) gives  
%\begin{eqnarray*}
$
2a_0 +(\mathbf{n}\cdot \mathbf{m})a_0^2 
\geq 2 \eta_1 \eta_2. 
$
%\end{eqnarray*}
Thus we obtain $a_0 \geq \eta_1 \eta_2$. 
On the other hand, for $\mathbf{n}\cdot \mathbf{m}<0$, 
(\ref{ineqa0}) gives 
%\begin{eqnarray*}
$
2a_0 +(\mathbf{n}\cdot \mathbf{m})a_0^2 
\leq 2 \eta_1 \eta_2,  
$
%\end{eqnarray*}
and $a_0 \leq \eta_1 \eta_2$. 
Therefore we conlude 
\begin{eqnarray*}
a_0 = \eta_1 \eta_2. 
\end{eqnarray*}
%%%%%%%%%%%
%%%%%%%%%%%
%%%%%%%%%%%%%%
Next, let us set $\mathbf{r}$ to be parallel with $\mathbf{a}$ so that we can write  
$\mathbf{a}$ as $\mathbf{a}=|\mathbf{a}| \mathbf{r}$. 
Then (\ref{IneqFund}) shows  
\begin{eqnarray*}
1+
2\eta_1\eta_2 (\mathbf{n}\cdot \mathbf{m})
+ (\eta_1 \eta_2)^2 
(\mathbf{n} \cdot \mathbf{m})^2 
\geq 
1+ 2(\mathbf{n}\cdot \mathbf{m}) \eta_1 \eta_2
+ (\mathbf{n} \cdot \mathbf{m})^2 
|\mathbf{a}|^2 
+ 2(\mathbf{n}\cdot \mathbf{m})
|\mathbf{a}|
\sqrt{1+ 2(\mathbf{n}\cdot \mathbf{m}) \eta_1 \eta_2}. 
\end{eqnarray*}
It gives for $\mathbf{n}\cdot \mathbf{m}\neq 0$,  
\begin{eqnarray}
(\eta_1 \eta_2)^2 \geq 
|\mathbf{a}|^2 
+ 2\frac{|\mathbf{a}|}{(\mathbf{n}\cdot 
\mathbf{m})}
\sqrt{1 + 2(\mathbf{n}\cdot \mathbf{m})\eta_1 \eta_2}.
\label{nmineq}
\end{eqnarray}
%A geometric argument shows that 
%$\mathbf{n}\cdot \mathbf{m}$ can take 
%an arbitrarily small positive value. 
We examine $\mathbf{n}\cdot \mathbf{m}$. 
For $\mathbf{a}=\mathbf{0}$, 
the condition $\mathbf{a}=|\mathbf{a}|\mathbf{r}$
does not give any restriction on $\mathbf{n}\cdot \mathbf{m}$. 
For $\mathbf{a} \neq \mathbf{0}$, 
let us  
%In fact, let us 
write $\mathbf{a}$ 
as $\mathbf{a}= |\mathbf{a}|\mathbf{e}$ 
with a normalized vector $\mathbf{e}$. 
We introduce a normalized vector $\mathbf{e}^{\perp}$
satisfying $\mathbf{e}\cdot \mathbf{e}^{\perp}=0$. 
Let us assume $\eta_1 \leq \eta_2$. 
(The case $\eta_1 \geq \eta_2$ can be 
treated similarly.) 
We set $\mathbf{n}
= \cos \theta_1 
\mathbf{e}+ \sin \theta_1 \mathbf{e}^{\perp}$. 
Then choosing $\mathbf{m}$ 
so that 
$\mathbf{m}= \cos \theta_2 \mathbf{e}
- \sin \theta_2 \mathbf{e}^{\perp}$ 
satisfies 
\begin{eqnarray}
\sin \theta_2 = \frac{\eta_1}{\eta_2}\sin \theta_1,
\label{sintheta}
\end{eqnarray} 
we obtain 
$\eta_1\mathbf{n}+ \eta_2 \mathbf{m}
= ( \eta_1 \cos \theta_1+\eta_2 \cos \theta_2)\mathbf{e}$.   
For $\theta_1 =0$, $\theta_2=0$ satisfies 
(\ref{sintheta}).
 We increase $\theta_1$ from $0$ to 
$\pi/2$ to obtain its corresponding $\theta_2$ 
monotoneously increasing from $0$ to $
\mbox{Arcsin} \left( \frac{\eta_1}{\eta_2}\right)$
and monotoneuously decreasing $\mathbf{n}\cdot \mathbf{m}
= \cos (\theta_1 + \theta_2)$ from $1$ to 
$\cos\left(\frac{\pi}{2}+ \mbox{Arcsin} \left(
\frac{\eta_1}{\eta_2}\right)\right) \leq 0$. 
Thus we find that 
$\mathbf{n}\cdot \mathbf{m}$ can take an 
arbitrarily small positive value. 
This concludes that (\ref{nmineq}) holds 
only if $|\mathbf{a}|=0$. 
Thus we find $N=\eta_1 \eta_2 \id$ and 
%%%%%%%%%%%%
\begin{eqnarray*}
\Lambda\left( \frac{1}{2}(\id + \mathbf{n}\cdot \boldsymbol\sigma)
\otimes \frac{1}{2}(\id + \mathbf{m}\cdot \boldsymbol\sigma)\right)
= \frac{1}{4}
\left((1+  \eta_1 \eta_2 (\mathbf{n} \cdot \mathbf{m}))
\id + (\eta_1 \mathbf{n} + \eta_2 \mathbf{m})\cdot \boldsymbol\sigma 
\right).
\end{eqnarray*} 
%%%%%%%%%%%%%%%%%%
This $\Lambda$ 
is a map from $\mathcal{L}_+(\mathbb{C}^2)\otimes_{min}
\mathcal{L}_+(\mathbb{C}^2)$ to $\mathcal{L}_+(\mathbb{C}^2)$. 
Thus we proved that $\Lambda_{\eta_1}$ and $\Lambda_{\eta_2}$ 
are min-tensor-compatible for $\eta_1^2 + \eta_2^2=1$ and their joint map is 
uniquely determined. 
\par
For $\eta_1^2+ \eta_2^2 <1$, 
we first consider $\Lambda_{\eta_1}$ and 
$\Lambda_{\sqrt{1-\eta_1^2}}$ whose joint map is 
denoted by $\Lambda$. 
As $\eta_2 < \sqrt{1-\eta_1^2}$ holds, 
$\Lambda_{\eta_2}$ can be written as 
$\Lambda_{\eta_2} = \Lambda_{\sqrt{1-\eta_1^2}}\circ 
\mathcal{E}$ with some channel $\mathcal{E}$. 
In fact, in general for $\eta'_2 < \eta'_1$, 
\begin{eqnarray*}
\Lambda_{\eta'_2}=\Lambda_{\eta'_1} \circ \Lambda_{\eta'_2/\eta'_1}
\end{eqnarray*}
holds. 
The map $\Lambda\circ (id \otimes \mathcal{E})$ 
is a joint map of $\Lambda_{\eta_1}$ and $\Lambda_{\eta_2}$\cite{
Plavala}. 
Thus $\Lambda_{\eta_1}$ and $\Lambda_{\eta_2}$ 
are min-tensor-compatible. 
%\cite{incompchannel, maxcomp}.
\end{proof}
%%%%%%%%%%%%%
%%%%%%%%%%%%%%
The above theorem shows that the set of 
all min-tensor-compatible pairs of channels is strictly 
larger than that of all quantum compatible pairs. 
For instance, the above theorem shows that  
$\Lambda_{\frac{1}{\sqrt{2}}}$ and 
$\Lambda_{\frac{1}{\sqrt{2}}}$ are 
min-tensor-compatible while 
they are quantum incompatible as 
$\eta_1 = \eta_2 = \frac{1}{\sqrt{2}}$ 
does not satisfy (\ref{eq:qoptimal}).
%%%%%%%%%%%
%%%%%%%%%%%%%%%%%%
%%%%%%%%%%%%%%%%%%%%%%%%%%%
%%%%%%%%%%%%  New Section %%%%%%%%%%%%%%%%%%%%%%%%%%
%%%%%%%%%%%%%%%%%%%%%%%%%%%
%%%%%%%%%%%%%%%%%%%
\subsection{min-tensor-compatibility vs almost quantum compatibility}
We introduced the notion of min-tensor-compatibility. 
Theorem \ref{theorem1} showed that 
there are pairs of min-tensor-compatible channels which are not 
quantum compatible. 
Joint maps of such pairs are not completely positive. 
%%%%%%
\par
Let aside the incompatibility issue for the moment, we consider a map 
$\Phi: \mathcal{L}(\mathbb{C}^2 \otimes \mathbb{C}^2) 
\to \mathcal{L}(\mathbb{C}^2)$ defined by 
%%%%%%%%%
%%%%%%%%%
\begin{eqnarray*}
\Phi(X) = \langle \varphi_{-}
 |(id \otimes T)(X)|\varphi_-\rangle \id, 
\end{eqnarray*}
where $|\varphi_-\rangle$ is a singlet state defined by 
$|\varphi_-\rangle = \frac{1}{\sqrt{2}}
(|01\rangle - |10\rangle)$  
and $T$ is a transpose map for $z$-basis.
For $A\geq 0$ and $B\geq 0$, 
$\Phi(A\otimes B) = \langle \varphi_- |
A\otimes T(B)|\varphi_-\rangle \id\geq 0$ holds 
as $T(B) \geq 0$. 
On the other hand,  
for $X= |\varphi_+\rangle \langle \varphi_+ |$ 
with $|\varphi_+ \rangle:= \frac{1}{\sqrt{2}}
(|00\rangle + |11\rangle)$, we obtain 
$\Phi(X)=-\id/2 <0$. 
Thus $\Phi$ is not completely 
positive but is well-defined as a map 
from $\mathcal{L}_+(\mathbb{C}^2)\otimes_{min}
\mathcal{L}_+(\mathbb{C}^2)$ to $\mathcal{L}_+
(\mathbb{C}^2)$. 
The map is unphysical as it is not completely positive. 
On the other hand, in the context of compatibility 
we are not interested in effects like 
$X= |\varphi_+\rangle \langle \varphi_+|$.
%Instead all the effects we concern have the form like $A\otimes B$.
Instead, quantities $P(A, B| \rho):= \mbox{tr}[\rho \Phi(A \otimes B)]$ 
for all $A, B \in \mathcal{L}_+(\mathbb{C}^2)$ and $\rho \in \mathcal{S}(\mathbb{C}^2)$ 
are all that we concern about in order to check 
if the map $\Phi$ is a joint map of some channels. 
Although $\Phi$ is not completely positive, 
the set of data $\{P(A,B|\rho)|\ 
0 \leq A, B \leq \id\}$
is physically realizable by the following 
procedure. First we map a state $\rho\in \mathcal{S}(\mathbb{C}^2)$ to 
a state $\Psi^*(\rho)$ on the composite system by a channel $\Psi$ defined 
by $\Psi(\cdot ) := \langle \varphi_- | \cdot |\varphi_-\rangle \id$. 
($\Psi^*(\rho) = |\varphi_-\rangle \langle \varphi_-|$ holds for all $\rho$. )
Then we measure an effect $A\otimes T(B) \geq 0$ 
and read its outcome probability as that for $A \otimes B$. 
More concretely, 
an apparatus measuring $T(B)$ is constructed by 
changing its reference frame from $\{\mathbf{e}_x, \mathbf{e}_y, 
\mathbf{e}_z\}$ to $\{\mathbf{e}_x, - \mathbf{e}_y, \mathbf{e}_z\}$.  
 Thus we may introduce yet another class of compatibility. 
 %%%%%%%
 \begin{definition}
 Consider a pair of channels 
 $\Lambda_1:\mathcal{L}(\hi_1) \to \mathcal{L}(\hi_0)$ and 
 $\Lambda_2: \mathcal{L}(\hi_2) \to \mathcal{L}(\hi_0)$. 
 The pair is called {\bf almost quantum compatible} if and only if 
 it is min-tensor-compatible and there exists a joint map $\Lambda$ 
 written in a form 
 \begin{eqnarray*}
 \Lambda = \Psi\circ (\Theta_1 \otimes \Theta_2), 
 \end{eqnarray*}
 where $\Psi: \mathcal{L}(\hi_1 \otimes \hi_2)\to 
 \mathcal{L}(\hi_0)$ is a channel, and 
 $\Theta_1: \mathcal{L}(\hi_1) \to \mathcal{L}(\hi_1)$ 
 and $\Theta_1:\mathcal{L}(\hi_2) \to \mathcal{L}(\hi_2)$ 
 are unit-preserving positive maps.  
 \end{definition}
 %%%%%%%%%%%
By definition almost quantum compatible pairs of channels are 
min-tensor-compatible. 
We show the converse is not true. 
There are min-tensor-compatible pairs that are not almost 
quantum compatible. 
%We ask if all the max-compatible pairs are almost quantum %compatible. 
%%%%%%
\begin{theorem}
Let $\Lambda_{\eta}$ $(0 \leq \eta \leq 1)$ be a noisy identity channel on a qubit. 
$\Lambda_{\eta_1}$ and $\Lambda_{\eta_2}$ are almost quantum compatible 
if and only if 
\begin{eqnarray}
\eta_1^2 + \eta_2^2 + (1-\eta_1- \eta_2)^2 \leq 1.  \label{eqaqc}
\end{eqnarray}
\end{theorem}
\begin{proof}
If part is trivial as $\Lambda_{\eta_1}$ and 
$\Lambda_{\eta_2}$ satisfying inequality (\ref{eqaqc}) are quantum compatible. 
\\
Let us consider the only if part. 
Suppose that 
$\Lambda_{\eta_1}$ and $\Lambda_{\eta_2}$ are almost quantum compatible.  
We assume $\eta_1+ \eta_2 \geq 1$ 
and show that  $\eta_1$ and $\eta_2$ satisfy 
(\ref{eqaqc}). 
There exists a channel $\Psi: \mathcal{L}(\mathbb{C}^2 
\otimes \mathbb{C}^2) \to \mathcal{L}(\mathbb{C}^2)$ and
unit-preserving
positive maps $\Theta_1, \Theta_2: \mathcal{L}(
\mathbb{C}^2) \to \mathcal{L}(\mathbb{C}^2)$ such that 
$\Lambda:= \Psi\circ (\Theta_1 \otimes \Theta_2)$ is a joint map 
of $\Lambda_{\eta_1}$ and $\Lambda_{\eta_2}$. 
We introduce an auxiliary qubit $\hik=\mathbb{C}^2$ and consider a singlet state 
$|\varphi_-\rangle 
\langle \varphi_- |
 \in \mathcal{S}(\hik \otimes \hi_0)= \mathcal{S}(
\mathbb{C}^2\otimes \mathbb{C}^2)$. 
While $(id\otimes \Lambda)^*(|\varphi_-\rangle 
\langle \varphi_-|)$ may not be a quantum state, 
$\rho:=(id \otimes \Psi)^*(|\varphi_- \rangle 
\langle \varphi_- |)$ is a quantum state on a composite system 
$\hik\otimes \hi_1 \otimes \hi_2=
\mathbb{C}^2 \otimes \mathbb{C}^2 \otimes \mathbb{C}^2$. 
(See \cite{Erkka} for a related technique.)
We examine tripartite correlation functions with respect to the state $\rho$. 
We study $\Theta_1$ 
to choose appropriate effects in considering the tripartite correlations. 
The map can be written with coefficients $v_i, A_{ij} \in \mathbb{R}^3$ $(i, j=1,2,3)$ 
as 
\begin{eqnarray*}
\Theta_1(\id)=\id,\quad 
\Theta_1(\sigma_i) = v_i \id + \sum_{j=1}^3 A_{ij} \sigma_j. 
\end{eqnarray*}
Thus $\mathbf{a}=[a_1, a_2, a_3] \in \mathbb{R}^3$ 
satisfying $\mathbf{a}\cdot \mathbf{v}=0$ with  
$\mathbf{v}=[v_1, v_2, v_3]$ gives 
\begin{eqnarray*}
\Theta_1(\mathbf{a}\cdot \boldsymbol\sigma) = \sum_{i,j=1}^3  
a_i A_{ij} \sigma_j. 
\end{eqnarray*}
That is, $\Theta_1(\mathbf{a}\cdot \boldsymbol\sigma)$
 has no $\id$ component. 
We introduce a two-dimensional subspace 
$V_1:=\{\mathbf{a}\in \mathbb{R}^3 |\ 
\mathbf{a} \cdot \mathbf{v}=0\}$. 
The above argument shows that
$\Theta_1 (\mathbf{a}\cdot \boldsymbol\sigma)$ 
 for 
any $\mathbf{a} \in V_1$ does not have 
$\id$ component. 
 %%%%%%
 %%%%%%%%
%%%%%%%%%%%%%
($V_1$ is $\mathbb{R}^3$ in case $\mathbf{v}=0$. The following argument holds 
also in this case.) 
By repeating the argument for $\Theta_2$, 
one can conclude that there is a two-dimensional subspace 
$V_2 \subset \mathbb{R}^3$ such that 
$\Theta_2(\mathbf{a}\cdot \boldsymbol\sigma)$ 
has no $\id$ coefficient for $\mathbf{a}\in V_2$. 
We denote by $V=V_1 \cap V_2$, which satisfies 
$\dim V\geq 1$.
%%%%%%%%%%%%%%
%%%%%%%%%
Let us choose normalized vectors 
$\mathbf{e}_2 \in V$ and $\mathbf{e}_1 \in V_1$ 
so that $\mathbf{e}_2\cdot \mathbf{e}_1 =0$. 
%%%%%%%%%
We write $\Theta_1(\mathbf{e}_1\cdot \boldsymbol\sigma)$
as
\begin{eqnarray*}
\Theta_1(\mathbf{e}_1\cdot \boldsymbol\sigma) &=& x_1 \mathbf{p}\cdot \boldsymbol\sigma, 
\end{eqnarray*}
where $\mathbf{p}$ is a normalized vector and 
$x_1 \in \mathbb{R}$. 
%where a normalized vector 
%$\mathbf{p}$ is chosen so that $x_1\geq 0$. 
As $\Theta_1$ is a 
unit-preserving positive map, $\Theta_1(\id \pm \mathbf{e}_1
\cdot \boldsymbol\sigma)
= \id \pm x_1\mathbf{p}\cdot \boldsymbol\sigma\geq 0$ holds, 
and hence $|x_1|\leq 1$ follows.  
%Since $\mathbf{e}_2 \in V\subset V_2$, 
$\Theta_1(\mathbf{e}_2 \cdot \boldsymbol\sigma)$ 
is written as 
$\Theta_1(\mathbf{e}_2 \cdot \boldsymbol\sigma)
=y_1 \mathbf{p}'\cdot \boldsymbol\sigma$, 
where $|y_1|\leq 1$ and $\mathbf{p}'$ is a 
normalized vector. 
We expand $\mathbf{p}'$ by introducing 
a normalized vector $\mathbf{q}$ satisfying 
$\mathbf{p}\cdot \mathbf{q}=0$ to obtain 
\begin{eqnarray*}
\Theta_1(\mathbf{e}_2\cdot \boldsymbol\sigma ) &=& y_1 (\sin \theta_1 \mathbf{p} + \cos \theta_1 \mathbf{q} )\cdot 
\boldsymbol\sigma,  
\end{eqnarray*}
where $\theta_1 \in \mathbb{R}$. 
We choose the direction of $\mathbf{q}$ 
so that $y_1 \cos \theta_1 \leq 0$. 
%%%%%%%%%%
Now we set a normalized vector $\mathbf{e}_3$ 
which is orthogonal to both $\mathbf{e}_1$ 
and $\mathbf{e}_2$. 
%so that $\mathbf{e}_1 \cdot \mathbf{e}_3
%= \mathbf{e}_2 \cdot \mathbf{e}_3=0$. 
$\Theta_2(\mathbf{e}_3\cdot 
\boldsymbol\sigma)$ is written as 
\begin{eqnarray*}
\Theta_2(\mathbf{e}_3\cdot \boldsymbol\sigma)&=&z'_2 \id + z_2 \mathbf{r}\cdot \boldsymbol\sigma, 
\end{eqnarray*}
where $z'_2, z_2 \in \mathbb{R}$ and 
$\mathbf{r}$ is a normalized vector. 
We choose the direction of $\mathbf{e}_3$ 
so that $z'_2\geq 0$. 
As $\Theta_2(\id \pm \mathbf{e}_3\cdot 
\boldsymbol\sigma)
\geq 0$ holds, 
we obtain $1-z'_2 \geq |z_2|$ and thus 
$0\leq z'_2 \leq 1$ and $|z_2|\leq 1$. 
Finally, we represent 
$\Theta_2(\mathbf{e}_2 \cdot 
\boldsymbol\sigma)$ as
\begin{eqnarray*}
\Theta_2(\mathbf{e}_2\cdot \boldsymbol\sigma)&=& y_2(\sin \theta_2 \mathbf{r} + \cos \theta_2 \mathbf{s})\cdot \boldsymbol\sigma,  
\end{eqnarray*}
where 
$\mathbf{s}$ is a normalized vector satisfying 
$\mathbf{r}\cdot \mathbf{s}=0$. 
The direction of $\mathbf{s}$ 
is chosen so that $y_2 \cos \theta_2 \leq 0$. 
$|y_2|\leq 1$ follows by considering 
$\Theta_2(\id \pm \mathbf{e}_2\cdot 
\boldsymbol\sigma) \geq 0$.
Note that none of $x_1, y_1, z_2$ and $y_2$ 
can be vanishing. For instance, if $x_1=0$, 
any channel $\Psi$ gives 
$\Psi\circ(\Theta_1 \otimes \Theta_2)
(\mathbf{e}_1 \cdot \boldsymbol\sigma\otimes \id)
=0$ which does not coincide with 
$\Lambda_{\eta_1}(\mathbf{e}_1 \cdot 
\boldsymbol\sigma)$. Similar argument holds for 
$y_1, z_2$ and $y_2$.   
%%%%%%%%%%%
%%%%%%%%%%%%
\par
We need the following lemma on Clifford algebra.
%%%%%%%%%%%%% 
\begin{lemma}\label{lemmaClifford}
Let $\{E_n\}_{n=1}^N$ be self-adjoint operators satisfying 
for any $n,m$,  
\begin{eqnarray*}
&&E_n E_m + E_m E_n=2 \delta_{nm} \id. 
\end{eqnarray*}
For $\mathbf{x}\in \mathbb{R}^N$, 
we introduce $E(\mathbf{x}):=\sum_{n=1}^N x_n E_n$.  
For any state $\rho$, it holds that 
\begin{eqnarray*}
| \mbox{tr}[\rho E(\mathbf{x})]| \leq |\mathbf{x}|. 
\end{eqnarray*}
In addition, it holds that 
\begin{eqnarray*}
\sum_n \mbox{tr}[\rho E_n]^2 \leq 1. 
\end{eqnarray*}
\end{lemma}
%%%%%%%
%%%%%%%%
The proof will be presented in Appendix \ref{app:A}. 
\par
We apply the Lemma to an anticommuting set
$\{\mathbf{e}_1\cdot \boldsymbol\sigma \otimes \mathbf{p}\cdot \boldsymbol\sigma 
\otimes \id,
\mathbf{e}_3 \cdot \boldsymbol\sigma \otimes 
\id \otimes \mathbf{r}\cdot \boldsymbol\sigma,
\id \otimes 
\mathbf{q}\cdot \boldsymbol\sigma \otimes \mathbf{s}\cdot \boldsymbol\sigma \}$
to obtain
\begin{eqnarray}
\mbox{tr}[\rho( \mathbf{e}_1\cdot \boldsymbol\sigma \otimes \mathbf{p}\cdot \boldsymbol\sigma 
\otimes \id )]^2
+ \mbox{tr}[\rho(\mathbf{e}_3 \cdot \boldsymbol\sigma \otimes 
\id \otimes \mathbf{r}\cdot \boldsymbol\sigma )]^2
+ \mbox{tr}[\rho(\id \otimes 
\mathbf{q}\cdot \boldsymbol\sigma \otimes \mathbf{s}\cdot \boldsymbol\sigma )]^2 \leq 1.
\label{ineq1} 
\end{eqnarray}
The first term of the left-hand side of the above inequality is calculated as
\begin{eqnarray}
\mbox{tr}[\rho( \mathbf{e}_1\cdot \boldsymbol\sigma\otimes \mathbf{p}\cdot \boldsymbol\sigma 
\otimes \id )]^2
&=& \frac{1}{x_1^2} \mbox{tr}[\rho(\mathbf{e}_1 \cdot \boldsymbol\sigma\otimes
\Theta_1(\mathbf{e}_1\cdot 
\boldsymbol\sigma) \otimes \id  )]^2
\nonumber 
\\
&=&\frac{1}{x_1^2}
\langle \varphi_- | \mathbf{e}_1 \cdot \boldsymbol\sigma \otimes \Psi (\Theta_1 ( \mathbf{e}_1 \cdot \boldsymbol\sigma)
\otimes \Theta_2(\id))  |\varphi_-\rangle^2 
\nonumber \\
&=& \frac{1}{x_1^2} 
\langle \varphi_- |
\mathbf{e}_1\cdot \boldsymbol\sigma \otimes 
\Lambda_{\eta_1} (\mathbf{e}_1 \cdot \boldsymbol\sigma) |\varphi_- \rangle^2
= \frac{\eta_1^2}{x_1^2}.
\label{eq1st}
\end{eqnarray}
The second term of the left-hand side of (\ref{ineq1}) is 
\begin{eqnarray}
\mbox{tr}[\rho(\mathbf{e}_3 \cdot \boldsymbol\sigma 
\otimes \id \otimes \mathbf{r}\cdot \boldsymbol\sigma ) ]^2
&=& \frac{1}{z_2^2}\mbox{tr}[\rho(\mathbf{e}_3 \cdot 
\boldsymbol\sigma 
\otimes \id \otimes 
\Theta_2(\mathbf{e}_3\cdot \boldsymbol\sigma - 
z'_2 \id 
))]^2
\nonumber \\
& =& 
\frac{1}{z_2^2} \langle \varphi_- | \mathbf{e}_3 \cdot 
\boldsymbol\sigma \otimes
\Lambda_{\eta_2}
(\mathbf{e}_3\cdot \boldsymbol\sigma - 
z'_2 \id 
)
|\varphi_- \rangle^2 =\frac{\eta_2^2}{z_2^2}.
\label{eq2nd} 
\end{eqnarray} 
Now we employ the following lemma to evaluate the third term of 
(\ref{ineq1}).  
%%%%%%%%%%%%%%%%%%
\begin{lemma}\label{lemmacorrelation}
Let $\{p(a,b,c)\}_{a,b,c\in \{0,1\}}$ be a probability distribution. 
We introduce the following quantities:
\begin{eqnarray*}
p_{12}:= \sum_{abc} p(a,b,c) (-1)^{a\oplus b}\\
p_{23}:= \sum_{abc} p(a,b,c) (-1)^{b \oplus c}\\
p_{13}:= \sum_{abc} p(a,b,c)(-1)^{a\oplus c}. 
\end{eqnarray*}
They satisfy 
\begin{eqnarray*}
p_{23} \geq p_{12}+p_{13} -1. 
\end{eqnarray*}
\end{lemma}
The proof is found in Appendix \ref{app:B}. 
%%%%%%
This lemma claims that the second and the third random variables are 
strongly correlated if the correlation between the first and the second 
and that between the first and the third are strong. 
We examine 
\begin{eqnarray*}
p(a,b,c):= \mbox{tr}\left[\rho 
\left(\frac{\id +(-1)^a \mathbf{e}_2\cdot \boldsymbol\sigma}{2}
\otimes \frac{\id + (-1)^b \mathbf{q}\cdot \boldsymbol\sigma }{2}
\otimes \frac{\id + (-1)^c \mathbf{s}\cdot \boldsymbol\sigma} {2}
\right)\right].
\end{eqnarray*}
For this probability distribution, one can calculate to obtain 
\begin{eqnarray*}
p_{12}=\mbox{tr}[\rho
(\mathbf{e}_2 \cdot \boldsymbol\sigma \otimes \mathbf{q}\cdot \boldsymbol\sigma \otimes \id)]
=\langle \varphi_-| 
\left(\mathbf{e}_2 \cdot \boldsymbol\sigma \otimes 
\Psi( \mathbf{q}\cdot \boldsymbol\sigma \otimes \id 
\right)|\varphi_- \rangle
\end{eqnarray*}
As it holds that 
\begin{eqnarray*}
\mathbf{q}\cdot \boldsymbol\sigma=\frac{1}{y_1 \cos \theta_1}
\Theta_1(\mathbf{e}_2 \cdot \boldsymbol\sigma) 
- \tan \theta_1 \mathbf{p}\cdot \boldsymbol\sigma 
= \frac{1}{y_1 \cos \theta_1}
\Theta_1(\mathbf{e}_2 \cdot \boldsymbol\sigma) 
- \frac{\tan \theta_1}{x_1} \Theta_1(\mathbf{e}_1 \cdot 
\boldsymbol\sigma ),  
\end{eqnarray*}
we obtain 
\begin{eqnarray*}
p_{12}
&=&
\langle \varphi_-| 
\left(\mathbf{e}_2 \cdot \boldsymbol\sigma \otimes 
\Psi( 
\Theta_1 \left(
\frac{\mathbf{e}_2 \cdot \boldsymbol\sigma}{y_1 \cos \theta_1}
- \frac{\tan \theta_1 \mathbf{e}_1 \cdot \boldsymbol\sigma}{
x_1}
\right)
 \otimes \Theta_2(\id) 
\right)|\varphi_- \rangle
\\
&=& \langle \varphi_-| 
\left(\mathbf{e}_2 \cdot \boldsymbol\sigma \otimes 
\Lambda_{\eta_1}
\left(
\frac{\mathbf{e}_2 \cdot \boldsymbol\sigma}{y_1 \cos \theta_1}
- \frac{\tan \theta_1 \mathbf{e}_1 \cdot \boldsymbol\sigma}{
x_1}
\right)\right)
|\varphi_- \rangle
=
- \frac{\eta_1}{y_1 \cos \theta_1}. 
\end{eqnarray*}
Similarly, as 
\begin{eqnarray*}
\mathbf{s}\cdot \boldsymbol\sigma 
&=&\frac{1}{y_2 \cos \theta_2}
\Theta_2(\mathbf{e}_2 \cdot \boldsymbol\sigma) - 
\tan \theta_2 \mathbf{r}\cdot \boldsymbol\sigma
\\
&=& \frac{1}{y_2 \cos \theta_2}
\Theta_2(\mathbf{e}_2 \cdot \boldsymbol\sigma) - 
\frac{\tan \theta_2}{z_2}
 \Theta_2 (\mathbf{e}_3 \cdot \boldsymbol\sigma -z'_2 \id)
\end{eqnarray*}
holds, we obtain 
\begin{eqnarray*}
p_{13}= -\frac{\eta_2}{y_2 \cos \theta_2}. 
\end{eqnarray*}
The condition 
$0 \geq y_1 \cos \theta_1, y_2 \cos \theta_2 \geq -1$ 
gives $p_{12} \geq \eta_1$ and $p_{13}\geq \eta_2$. 
Applying Lemma \ref{lemmacorrelation}, 
we obtain 
\begin{eqnarray}
p_{23}
= \mbox{tr}[\rho (\id \otimes \mathbf{q}\cdot 
\boldsymbol\sigma \otimes \mathbf{s}
\cdot \boldsymbol\sigma)]
 \geq \eta_1 + \eta_2 -1.
\label{eqp23} 
\end{eqnarray} 
Combining (\ref{eqp23}), (\ref{eq1st}), (\ref{eq2nd})
 and (\ref{ineq1}), we find 
\begin{eqnarray*}
\frac{\eta_1^2}{x_1^2}+ \frac{\eta_2^2}{z_2^2} 
+ (\eta_1 + \eta_2 -1)^2\leq 1.
\end{eqnarray*}
As $|x_1|, |z_2|\leq 1$, we obtain the wanted inequality. 
\end{proof}
%%%%%%%%%%%%%%%%%%%5
%%%%%%%%%%%%
\section{Discussions}
%%%%%%%%%%
In the expanded concept of compatibility, a pair of maps is termed compatible 
if a joint map exists. The presence of this joint map might vary depending 
on how we define a composite system. This paper addresses this issue through 
an examination of composite systems involving two qubits. 
Min-tensor-compatibility 
refers to compatibility within the context of the min-tensor product of effect spaces. 
We established the necessary and sufficient condition for a pair of noisy identity 
channels on qubits, denoted as $\Lambda_{\eta_1}$ and $\Lambda_{\eta_2}$, to be min-tensor-compatible. 
These condition reveals the existence of min-tensor-compatible pairs of noisy identity 
channels that do not satisfy the criteria for quantum compatibility. 
This underscores that compatibility is contingent upon the approach to composite systems. 
This is distinct from the Tsirelson bound for the CHSH inequality, which yields 
an identical value for both normal quantum and min-tensored composite systems.
Additionally, we introduced a novel concept, almost quantum compatibility, 
driven by operational perspectives. We derived the necessary and sufficient condition
 for a pair of noisy identity channels on a qubit to be almost quantum compatible, 
which aligns with those for quantum compatibility. 
In proving the condition, we investigated tripartite correlation functions of quantum states.
 While any bipartite correlations for quantum effects can be realized by quantum states\cite{bipartite}, 
there exist tripartite correlations that cannot be replicated by any quantum states\cite{tripartite}. 
We introduced a new criterion in our proof to ascertain whether 
two given bipartite correlations can be extended to a tripartite quantum correlation.
\par 
There are still avenues for further exploration. Regarding the set of channel pairs, 
we demonstrated the inclusion: 
\begin{eqnarray*}
\mbox{quantum compatible} 
\subseteq \mbox{almost quantum compatible}
\subsetneq \mbox{min-tensor-compatible
}.
\end{eqnarray*}
Whether the sets of all quantum compatible pairs and all almost quantum compatible pairs 
coincide remains unknown. Additionally, from a device-independent perspective, 
we may introduce another class of compatibility allowing a channel 
$\Psi$ 
to have a domain as 
 $\mathcal{L}(\mathbb{C}^{d_1} 
\otimes \mathbb{C}^{d_2})$
and positive maps $\Theta_j$ $(j=1,2)$
to have codomains $\mathcal{L}(
\mathbb{C}^{d_j})$
with 
$d_1, d_2 \geq 2$.
Finally, the overarching goal is to develop a general theory unrestricted to qubits.
%%%%%%%%%%%%%%%%%
\begin{acknowledgments}
TM acknowledges financial support from JSPS (KAKENHI Grant No. 
JP20K03732).
\end{acknowledgments}
%%%%%%%%%%%%%%%%%%%%%%%%%%%%%

%%%%%%%%%%
\begin{appendix}
\section{Proof of Lemma \ref{lemmaClifford}}\label{app:A}
%%%
%%%%%%%%%
%%%%%%%%%%
\begin{proof}
Note that $E(\mathbf{x})^2 
= |\mathbf{x}|^2 \id$ holds. 
Thus we obtain 
\begin{eqnarray*}
|\mbox{tr}[\rho E(\mathbf{x})]|^2 
\leq \mbox{tr}[\rho E(\mathbf{x})^2]
= |\mathbf{x}|^2. 
\end{eqnarray*}
%%%
For any $\mathbf{x}\in \mathbb{R}^N$ 
with $|\mathbf{x}|= 1$ it holds that  
\begin{eqnarray*}
\sum_n x_n \mbox{tr}[\rho E_n] \leq 1. 
\end{eqnarray*}
Taking $x_n=\frac{\mbox{tr}[\rho E_n]}{\sqrt{\sum_m \mbox{tr}[\rho E_m]^2}}$, 
we obtain 
\begin{eqnarray*}
\sum_n \mbox{tr}[\rho E_n]^2 \leq 1.  
\end{eqnarray*}
\end{proof}
%%%%%%
\section{Proof of Lemma \ref{lemmacorrelation}}\label{app:B}
\begin{proof}
As $\sum_{c,a \oplus b=0}p(a,b,c)+ \sum_{c,a \oplus b=1}p(a,b,c)=1$, 
\begin{eqnarray*}
p_{12}= 
2 \sum_{c,a\oplus b=0}p(a,b,c) -1 
=2( p(0,0,0)+p(0,0,1) +p(1,1,0)+p(1,1,1))-1 
\end{eqnarray*}
holds. Similarly, we have 
\begin{eqnarray*}
p_{13}= 2 \sum_{b, a\oplus c=0}p(a,b,c)-1 
= 2( p(0,0,0) + p(0,1,0) +p(1,0,1)+p(1,1,1))-1. 
\end{eqnarray*}
Therefore we obtain 
\begin{eqnarray*}
1 + \frac{p_{12}+p_{13}}{2}& =& 2 (p(0,0,0) + p(1,1,1)) 
+ p(1,1,0) +p(0,0,1) + p(0,1,0) + p(1,0,1)
\\ 
&=& 1+ p(0,0,0) +p(1,1,1) - p(0,1,1) - p(1,0,0) \leq 1+ p(0,0,0)+p(1,1,1). 
\end{eqnarray*}

Thus 
\begin{eqnarray*}
p_{23}=2 ( p(0,0,0)+ p(1,0,0) +p(0,1,1)+p(1,1,1))-1 
\geq 2 (p(0,0,0) +p(1,1,1)) -1 
\geq p_{12}+p_{13} -1. 
\end{eqnarray*}

\end{proof}
\end{appendix}
\end{document}